\documentclass[a4paper, reqno]{amsart}

\usepackage{hyperref}
\usepackage{latexsym}
\usepackage{amsmath}
\usepackage{amssymb}
\usepackage{amsfonts}
\usepackage{amsthm}
\usepackage{amstext}
\usepackage{enumerate}
\usepackage{pdflscape}
\usepackage{graphicx}
\usepackage{float}
\usepackage{color}

      \newcommand{\Bog}{{\operatorname{Bog}}}
      \renewcommand{\i}{{\operatorname{i}}}
     \newcommand{\e}{\operatorname{e}}

     \newcommand{\ext}{{\operatorname{ext}}}

     \newcommand{\s}{{\operatorname{s}}}
     
     \renewcommand{\b}{{\operatorname{b}}}
     
     \renewcommand{\d}{{\operatorname{d}}}

     \newcommand{\Ran}{{\operatorname{Ran}}}

     \newcommand{\Z}{{\mathbb{Z}}}

\newcommand{\vecsp}{\overrightarrow{\operatorname{sp}}}

\newcommand{\zz}{{\mathbb{Z}}}
\newcommand{\cc}{{\mathbb{C}}}

\newcommand{\rr}{{\mathbb{R}}}
\newcommand{\0}{{\mathbf{0}}}
\newcommand{\x}{\mathbf{x}}

\newcommand{\q}{{\bf{q}}}
\newcommand{\p}{{\bf{p}}}
\renewcommand{\k}{{\bf k}}

\renewcommand{\p}{\mathbf{p}}

\newcommand{\nph}{{\rm nph}}

\newcommand{\cK}{{\mathcal K}}
\newcommand{\cH}{{\mathcal H}}
\newcommand{\hc}{{\rm hc}}

     \theoremstyle{plain}
     \newtheorem{thm}{Theorem}[section]
     
     \newtheorem{lemma}[thm]{Lemma}
     \newtheorem{cor}[thm]{Corollary}
     \theoremstyle{definition}

     \newtheorem{remark}[thm]{Remark}

\def\bbbone{{\mathchoice {\rm 1\mskip-4mu l} {\rm 1\mskip-4mu l}
{\rm 1\mskip-4.5mu l} {\rm 1\mskip-5mu l}}}
\def\one{\bbbone}
\newcommand{\beqnn}{\begin{eqnarray*}}
\newcommand{\eeqnn}{\end{eqnarray*}}
\newcommand{\beqn}{\begin{eqnarray}}
\newcommand{\eeqn}{\end{eqnarray}}
\newcommand{\beq}{\begin{equation}}
\newcommand{\eeq}{\end{equation}}

     \numberwithin{equation}{section}
\newpage
\begin{document}
\title{Excitation spectrum of interacting Bosons\\ in the mean-field infinite-volume limit}

\author{Jan Derezi\'{n}ski}
\address[J. Derezi\'{n}ski]
{Dept. of Math. Methods in Phys., Faculty of Physics, University of Warsaw\\ 
Hoza 74, 00-682 Warszawa, Poland} 
\email{Jan.Derezinski@fuw.edu.pl}

\author{Marcin Napi\'{o}rkowski}
\address[M. Napi\'{o}rkowski]
{Dept. of Math. Methods in Phys., Faculty of Physics, University of Warsaw\\ 
Hoza 74, 00-682 Warszawa, Poland} 
\email{Marcin.Napiorkowski@fuw.edu.pl}

\date{\today}

\begin{abstract} 
We consider  homogeneous Bose gas in a large cubic box with periodic
boundary conditions, at zero temperature. We analyze its excitation
spectrum in a certain kind of a mean field infinite volume
 limit. We prove that under appropriate conditions the excitation spectrum has the form predicted by the Bogoliubov approximation. Our result can be viewed as an extension of the result of Seiringer  \cite{S} to large volumes.
\end{abstract}
\maketitle

\section{Introduction and main results}
\label{s1}
Many physical properties of complicated interacting systems can be
derived from simple Hamiltonians invoving independent (bosonic or
fermionic) {\em quasiparticles} (see \cite{DNM} for a detailed discussion of
this concept) with appropriately chosen dispersion relation (the
dependence of the quasiparticle energy on the momentum). One of such
systems is the weakly interacting Bose gas at zero temperature.
 On the
heuristic level, the quasiparticle description of the Bose gas can be
derived from the Bogoliubov approximation (\cite{Bo}, see also \cite{CDZ}). The main goal
of this paper is a rigorous justification of this approximation 
for a homogeneous system of $N$ interacting bosons in a
certain kind of a mean field large volume limit.

Let us state the assumptions on the 2-body potential that we will use
throughout the paper. Consider a real function
$\mathbb{R}^{d}\ni \x\mapsto v
(\x)$, with its Fourier transform defined  by
\[\hat{v}(\p):=\int_{\mathbb{R}^{d}}v(\x)\e^{-\i\p\x} \d \x.\]
 We assume that $v(\x)=v(-\x)$, and that $v\in L^1(\mathbb{R}^{d})$ and $\hat v\in L^1(\mathbb{R}^{d})$. We also suppose  that the potential is positive and positive definite, i.e.
\[v(\x)\geq0,\ \ \x\in\mathbb{R}^d,\ \ \ \ \hat{v}(\p)\geq 0, \,\,\, \p \in \mathbb{R}^{d}.\]

We will consider Bose gas in large but finite  volume. To do this,
following the standard approach, we replace
the infinite space $\rr^d$ by the torus  $\Lambda=]-L/2,L/2]^{d}$,
    that is, the $d$-dimensional cubic box of side length $L$. We will
    always assume that $L\geq1$.

The original potential $v$ is replaced by its 
{\em periodized} version
\[v^{L}(\x):=\frac{1}{L^d}\sum_{\p\in(2\pi/L)\mathbb{Z}^{d}}\e^{\i\p\x}\hat{v}(\p).\]
Here, $\p\in(2\pi/L)\mathbb{Z}^{d}$ is the discrete momentum variable.
 Note that $v^{L}$ is periodic with respect to the domain $\Lambda$ and that $v^{L}(\x)\rightarrow v(\x)$ as $L\to \infty$. Consider the Hamiltonian 
\beqn
-\sum_{i=1}^{N}\Delta^{L}_{i}+\lambda\sum_{1\leq i<j\leq N} v^{L}(\x_{i}-\x_{j}) \label{hamiltonian1}
\eeqn  
acting on the space $L^{2}_{\s}(\Lambda^{N})$ (the symmetric subspace
of
 $L^{2}(\Lambda^{N})$)
The Laplacian is assumed to have periodic boundary conditions.

 Let $\rho=N/L^d$ be the density of the  gas. The Bogoliubov approximation \cite{Bo} predicts that the ground state energy is
\[\frac12\lambda\rho\hat v(\0)(N-1)-\frac{1}{2}\sum_{\p\in\frac{2\pi}{L}\mathbb{Z}^{d}\setminus \{ 0\}}\left(|\p|^{2}+\rho\lambda\hat{v}(\p)-|\p|\sqrt{|\p|^{2}+2\rho\lambda\hat{v}(\p)}\right) \]
and that the low lying excited states can  be derived from the
following {\em elementary excitation spectrum}:
\beq
|\p|\sqrt{|\p|^{2}+2\rho\lambda \hat{v}(\p)}. \label{bogol}\eeq

Note that within the Bogoliubov approximation
both the ground state energy and the excitation spectrum depend on $ \rho$ and $\lambda$ only through the product $\rho\lambda$. The dependence on $L$ is very weak:
\begin{enumerate}
\item The elementary excitation spectrum  (\ref{bogol})  depends on
  $L$ 
 only through the spacing of the momentum lattice $\frac{2\pi}{L}\zz^d$.
\item The expression for the ground state energy
divided by the volume $L^d$ converges for $L\to\infty$  to a finite expression
\beq \frac12\rho^2\lambda\hat v(\0) -\frac{1}{2(2\pi)^d}\int\left(|\p|^{2}+\rho\lambda\hat{v}(\p)-|\p|\sqrt{|\p|^{2}+2\rho\lambda\hat{v}(\p)}\right)\d\p. \label{hammi}\eeq
\end{enumerate}

 We believe that it is important to understand the Bogoliubov approximation for large $L$. Important physical properties, such as the phonon group velocity and the description of the Beliaev damping in terms of analyticity properties of Green's functions, have an elegant description when we can view the momentum as a continuous variable, which is equivalent to taking the limit $L\to\infty$.

Note that in our problem there are three {\em a priori} uncorrelated parameters: $\lambda$, $N$ and $L$. By the {\em mean field limit} one usually understands $N\to\infty$ with $\lambda\simeq\frac1N$ and $L=\rm{const.}$ However, when  both $N$ and $L$ are large it is natural to consider a somewhat different scaling. In our paper the mean field limit will correspond to $N\to\infty$ with $\lambda\simeq\frac1\rho=\frac{L^d}{N}$.

Motivated by the above argument we will consider a system described by the Hamiltonian
\beqn
H_N^{L}=-\sum_{i=1}^{N}\Delta^{L}_{i}+\frac{L^d}{N}\sum_{1\leq i<j\leq N} v^{L}(\x_{i}-\x_{j}). \label{hamiltonian}
\eeqn  
It is translation invariant -- it commutes with the total momentum operator
\beqn
P_N^L:=-\sum_{i=1}^{N}\i\partial_{\x_i}^L.
\eeqn

We will denote by $E_N^L$ the ground state energy of (\ref{hamiltonian}).
 If
 $\p\in\frac{2\pi}{L}\zz^d\backslash\{\0\}$ let 
$K_N^{L,1}(\p),K_N^{L,2}(\p),\dots$ be the eigenvalues of $H_N^L-E_N^L$
 of total momentum $\p$ in the order of increasing values, counting the multiplicity.
The lowest eigenvalue of $H_N^L-E_N^L$ of total momentum $\p=\0$ is $0$
by general arguments \cite{CDZ}. Let
$K_N^{L,1}(\0),K_N^{L,2}(\0),\dots$ be the next eigenvalues of
$H_N^L-E_N^L$ of total momentum $\0$, also in the order of increasing values,  counting the multiplicity.

 We also introduce the {\em Bogoliubov  energy}
\[ E_\Bog^{L}:=-\frac{1}{2}\sum_{\p\in\frac{2\pi}{L}\mathbb{Z}^{d}\setminus \{ 0\}}\left(|\p|^{2}+\hat{v}(\p)-|\p|\sqrt{|\p|^{2}+2\hat{v}(\p)}\right)\]
 and the {\em Bogoliubov elementary excitation spectrum}
\begin{gather}
e_{\p}=|\p|\sqrt{|\p|^{2}+2\hat{v}(\p)}. \label{ep}
\end{gather}
For any $\p\in\frac{2\pi}{L}\zz^d$ we consider the corresponding excitation energies with momentum $\p$:
\[
\Big\{\sum_{i=1}^j e_{\k_i}\ :\ 
\k_1,\dots,\k_j\in\frac{2\pi}{L}\zz^d\backslash\{\0\}, \ \k_1+\cdots+\k_j=\p,\ \ 
j=1,2,\dots\Big\}.\]
Let $K_\Bog^{L,1}(\p),K_\Bog^{L,2}(\p),\dots$ be these excitation
energies in the 
order of increasing values, counting the multiplicity. 
 We will use the term {\em excitation spectrum in the Bogoliubov approximation} to denote the set of pairs $\big(K_\Bog^{L,j}(\p),\p\big)\in\rr\times\rr^d$.
Later on we will see that it coincides with the joint spectrum
of commuting operators  $H_\Bog^L-E^L_\Bog$ and $P^L$
with  $(0,\0)$ removed. (See   \eqref{stan} for the definition of  $H_\Bog^L$.)

Below we
present pictures  of the excitation spectrum of 1-dimensional
Bose gas in the
Bogoliubov approximation for two potentials, $v_1$ and $v_2$. 
Both
potentials are appropriately scaled Gaussians. (Note that Gaussians 
satisfy the assumptions of our main theorem).
On both pictures  the (black) dot
 at the origin corresponds to the quasiparticle vacuum, (red) dots
 correspond to 1-quasiparticle excitations, (blue) triangles
 correspond to 2-quasiparticles excitations,
 while (green) squares correspond to $n$-quasiparticles excitations
 with $n\geq 3$.
We also give the graphs of the Fourier transforms of both potentials.

 Note that all figures are drawn in the same scale, apart from
 Fig. 4. where the potential had to be scaled down because of 
 space limitations. In our units of length 
 $\frac{2\pi}{L}=\frac{15}{100}$. 

 Notice that 
for some total momentum $\p$ many-quasiparticle excitation energies 
 are lower than the elementary excitation spectrum. In particular, for the potential $v_1$ it happens  already for  low
 momenta. 
Physically this means that  the corresponding 1-quasiparticle excitation    is not stable: it may decay to $m$-quasiparticle states, $m\geq 2$, with a lower energy. This phenomenon has been observed experimentally \cite{HMHF}
and is called the \textit{Beliaev damping} \cite{Be}. If one can assume that the momentum variable is continuous, the Beliaev damping corresponds to a pole of the Green's function on a non-physical sheet of the energy complex plane. The imaginary part of the position of this pole, computed by Beliaev, is responsible for the rate of  decay of quasiparticles.

\begin{figure}[H]
\begin{center}
\includegraphics[scale=0.75]{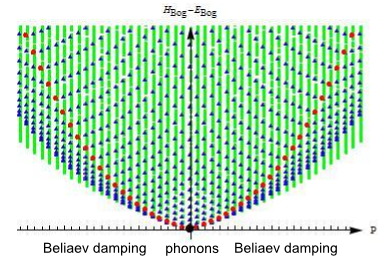}
\caption{Excitation spectrum of  1-dimensional 
homogeneous Bose gas with
potential
$v_1$ in the Bogoliubov approximation.}
\end{center}
\end{figure}
\begin{figure}[H]
\begin{center}
\includegraphics[scale=0.75]{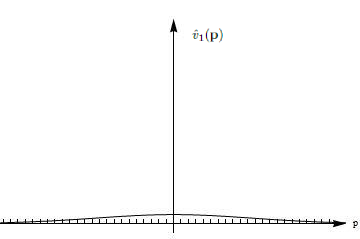}
\caption{$\hat{v}_1(\p)=\frac{\e^{-\p^2/5}}{10}$.} 
\end{center}
\end{figure}

The excitation spectrum for potential $v_2$ has a very different shape
-- it has  local maxima and  local minima away from the zero momentum.
 On the  picture we 
 show  traditional names of quasiparticles -- \textit{phonons} in the low momentum region, where the dispersion relation is
 approximately linear,  \textit{maxons} near the local maximum and
 \textit{rotons} near the local minimum of the elementary excitation
 spectrum (see \cite{G} for details). 
\begin{figure}[H]
\begin{center}
\includegraphics[scale=0.75]{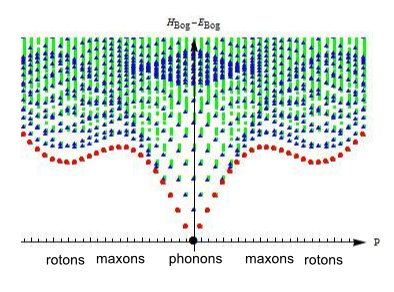}
\caption{Excitation spectrum of  1-dimensional 
homogeneous Bose gas with  potential $v_2$ in the Bogoliubov approximation.}
\end{center}
\end{figure}
\begin{figure}[H]
\begin{center}
\includegraphics[scale=0.75]{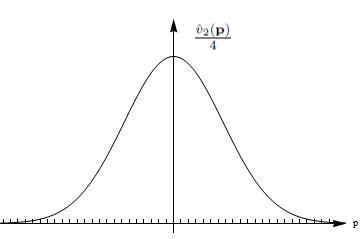}
\caption{$\hat{v}_2(\p)=\frac{15\e^{-\p^2/2}}{2}$.} 
\end{center}
\end{figure} 
 
From now on, we will drop the superscript $L$. Let us state our main
result. It is slightly different for the upper and lower bound:

\begin{thm}\label{thm}
\begin{enumerate}
\item Let $c>0$. Then there exists  $C$ such that
\begin{enumerate}
\item  if 
\beq L^{2d+2}\leq cN,\label{war1a}\eeq
 then 
\beq E_N\geq
\frac{1}{2}\hat{v}(\0)(N-1)+E_\Bog-CN^{-1/2}L^{2d+3}; \label{thmnier1}\eeq
\item if in addition 
\beq K_N^j(\p)\leq c NL^{-d-2},\label{war1b}\eeq then
\begin{eqnarray}
E_N+K_N^j(\p)&\geq&
\frac{1}{2}\hat{v}(\0)(N-1)+E_\Bog+K_\Bog^j(\p)\notag \\&&
-CN^{-1/2}L^{d/2+3}\big(K_N^j(\p)+L^d\big)^{3/2}.  \label{thmnier2}
\end{eqnarray}
\end{enumerate}
\item  Let $c>0$. Then there exists  $c_1>0$ and $C$ such that
\begin{enumerate}
\item if 
\begin{eqnarray}
 L^{2d+1}&\leq &cN \label{war2a}\\
\hbox{ and }\hspace{4ex} L^{d+1}&\leq &c_1 N,\label{war2a+}
\end{eqnarray}
 then
\beq E_N\leq
\frac{1}{2}\hat{v}(\0)(N-1)+E_\Bog+CN^{-1/2}L^{2d+3/2}; \label{thmnier3}\eeq
\item  if in addition 
\begin{eqnarray}
 K_\Bog^j(\p)&\leq &cN L^{-d-2}\label{war2b}\\
 \hbox{ and }\hspace{4ex} K_\Bog^j(\p)&\leq &
c_1NL^{-2},\label{war2b+}\end{eqnarray}
 then
\begin{eqnarray}
E_N+K_N^j(\p)&\leq&
\frac{1}{2}\hat{v}(\0)(N-1)+E_\Bog+K_\Bog^j(\p)\notag \\&&
+CN^{-1/2}L^{d/2+3}(K_\Bog^j(\p)+L^{d-1})^{3/2}.
\label{thmnier4}
\end{eqnarray}
\end{enumerate}
\end{enumerate}
\end{thm} 

 Let us stress that the constants $C$ and $c_1$ that appear in the theorem  depend only the potential $v$, the dimension $d$, and the constant $c$, but {\em do not} depend on $N$, $j$ and $L$. Note also that both in (1), resp. (2) we can deduce (a) from (b) by setting $K_N^j(\p)=0$, resp. $K_\Bog^j(\p)=0$.

Theorem \ref{thm}  expresses the idea that the Bogoliubov approximation becomes exact for large $N$ and $L$ provided that the volume does not grow too fast. This may appear not very transparent, since the error terms in the theorem depend on two parameters $L$ and $N$ as well as on the excitation energy. Therefore, we give some consequences of our theorem, where the error term depends only on $N$. They  generalize the corresponding remarks of \cite{S}.

\begin{cor}\label{coro}
Let $b>1$, $-1-\frac{1}{2d+1}< \alpha \leq 1$ and $L^{4d+6}\leq b N^{1-\alpha}$. Then there exists $M$ such that if $N>M$, then \vspace{2ex}
\begin{enumerate}
\item \hspace{7ex} $E_N=\frac{1}{2}\hat v(\0)(N-1)+E_\Bog+O(N^{-\alpha/2})$;
\vspace{2ex}
\item if $\min\left(K_N^j(\p),K_\Bog^j(\p)\right)\leq (bN^{1-\alpha}L^{-d-6})^{1/3}$, then
\[K_N^j(\p)= K_\Bog^j(\p)+O(N^{-\alpha/2});\]
\item 
  if  $0<\alpha\leq1$ and
 $\min\left(K_N^j(\p),K_\Bog^j(\p)\right)\leq bN^{1-\alpha}L^{-d-6}$, then
 \[ K_N^j(\p)= K_\Bog^j(\p)+\left(1+K_\Bog^j(\p)\right)O(N^{-\alpha/2}).\]
\end{enumerate}
\end{cor}

 The proof that Thm \ref{thm} implies  Cor. \ref{coro} is given in Appendix.

\begin{remark}
\begin{enumerate}
\item The case $\alpha=1$, $L=1$ of Corollary \ref{coro}
corresponds directly to the result of \cite{S}.
\item In part (3) of Corollary \ref{coro} one can also include the case $\alpha=0$ provided that $L$ is sufficiently large.
\end{enumerate}
\end{remark}

Thus, for large $N$ within a growing range of the volume, the low lying
energy-momentum spectrum of the homogeneous Bose gas is well described
by the Bogoliubov approximation. 
In the infinite volume limit  momentum 
 becomes a continuous variable, which is important when we want to consider the so-called \textit{critical velocity} and \textit{phase velocity} introduced by Landau. They play a crucial role in his theory of superfluidity (\cite{La1}, \cite{La2}, see also \cite{CDZ},\cite{ZB}).
 

 Mathematically, the Bogoliubov approximation has been studied mostly in the context of the ground state energy (\cite{LS}, \cite{LS1}, \cite{ESY}, \cite{GS}, \cite{So} \cite{YY}, see also \cite{LSSY}). This makes the work of Seiringer (\cite{S}), Grech-Seiringer (\cite{GrS}) and more recently by Lewin, Nam, Serfaty and Solovej (\cite{LNSS}) even more notable, since they are devoted to a rigorous study of the excitation spectrum of a Bose gas.

  In \cite{S} Seiringer proves that for a system of $N$ bosons on a flat unit torus $\mathbb{T}^{d}$ which interact with a  two-body interaction $v(\x)/(N-1)$, the excitation spectrum up to an energy $\kappa$ is formed by elementary excitations of momentum $\p$ with a corresponding energy of the form
(\ref{bogol})
up to an error term  of the order $O(\kappa^{3/2}N^{-1/2})$. Also in \cite{GrS} and \cite{LNSS} the authors are concerned with  finite systems in the large particle number limit.

Our result can be considered as an extension of Seiringer's result to
systems of arbitrary volume. The ultimate goal would be to prove
similar results in the thermodynamic limit with a fixed coupling
constant. Since this is at the moment out of reach, we try to pass to
some other limits, which  involve convergence of the volume to infinity.

The rest of this paper is devoted to  a 
 proof of Theorem \ref{thm}. It uses partly the methods presented in \cite{S}. Note, however, that naive mimicking leads  to a much weaker result, which involves assuming that $N\geq C \e^{cL^{d/2}}$ to ensure that the error terms tend to zero when taking the infinite volume limit. This can be easily seen by looking for example at equation (24) of \cite{S}. In this equation one of the constants is given by the expression $\e^{C_{2}}$ where $C_{2}$ is given by 
$$\sqrt{\frac{64N}{N-1}\sum_{\p\neq\0}\beta^{2}_{\p}}.$$
In the infinite volume limit the sum could be replaced by an integral which one can  compensate by the factor $L^{d/2}.$  This leads to a factor $\e^{c L^{d/2}}$ in the estimates.


Our proof uses certain identities that allow us to simplify the
algebraic computations involved in the proof. We use the method of
second quantization, working in the Fock space containing all
$N$-particle spaces at once. We embed this space in  the so-called
{\em extended space}, which contains nonphysical states with a negative
number of zero modes. This method  leads to relatively simple
algebraic calculations, which is helpful when we want to control the
volume dependence. Note also  that our 
method yields the same results as in \cite{S} if one takes $L=1$. 

 Strangely, we have never seen the method of the extended space in the literature. Some authors (starting with Bogoliubov in 1947, see \cite{Bo1}) introduce the operator $a_0^\dagger(\one+ N_0)^{-1/2}$, which coincides with our operator $U^\dagger$ on the physical space. Both operators increase the number of zeroth modes by one. The operator $U^\dagger$, however, acts on the extended space and is unitary, whereas $a_0^\dagger(\one+N_0)^{-1/2}$  acts on the physical space and is only isometric.

One can also see some similarity of our method with that of \cite{LNSS} where,  however, states with a negative number of modes do not appear.

\section{Miscellanea}

Let us describe some notation and basic facts from operator theory
used in our paper.

If $A$, $B$ are operators, then the following inequality will be often used:
\beq -A^\dagger A-B^\dagger B\leq A^\dagger B+B^\dagger A\leq A^\dagger A+B^\dagger B.\eeq

We will write $A+\hc$ for $A+A^\dagger $.

If $A$ is a self-adjoint operator and $\Omega$ a Borel subset of the spectrum of $A$, then $\one_\Omega(A)$ will denote the spectral projection of $A$ onto  $\Omega$.

Let $A$ be a bounded from below self-adjoint operator on  Hilbert space $\cH$.  For simplicity, let us assume that it has only discrete spectrum.
We define
\[\vecsp(A):=(E_1,E_2,\dots),\]
where $E_1,E_2,\dots$ are the eigenvalues of $A$ in the  order of
increasing values, counting the multiplicity. If $\dim\cH=n$, then we set $E_{n+1}=E_{n+2}=\dots=\infty$.

We will use repeatedly two consequences of the {\em min-max principle} \cite{RS}: 
\[A\leq B \ \hbox{ implies }\vecsp(A)\leq \vecsp(B),\]
and
the so-called {\em Rayleigh-Ritz principle}:
If $\cK$ is a closed subspace of $\cH$, let
 $P_\cK$ be the projection onto $\cK$. Then we have 
\beq
\vecsp(A)\leq\vecsp\Big(P_\cK AP_\cK\Big|_\cK\Big).\eeq

\section{Second quantization}

As discussed in the introduction, the main object of our paper, the Hamiltonian $H_N$ is defined
on the $N$-particle bosonic space
\[\cH_N:
=L^{2}_{\s}(\Lambda^{N}).\] We will work most of the time in the
momentum representation, in which the 1-particle space $L^2(\Lambda)$ 
is represented as
$l^2\bigl(\frac{2\pi}{L}\zz^d\bigr)$, thus
\[\cH_N\simeq
\otimes_{\s}^Nl^2\bigl(\frac{2\pi}{L}\zz^d\bigr).\] 

It is convenient to consider simultanously the direct sum of the
$N$-particle spaces, the {\em bosonic Fock space}
\begin{equation}
\mathcal{H}:=\mathop\oplus\limits_{N=0}^\infty\cH_N
=\Gamma_\s\Big(l^2\bigl(\frac{2\pi}{L}\zz^d\bigr)\Bigr).
\label{physicalspace1}
\end{equation}


The direct sum of the Hamiltonians $H_N$ 
will be denoted $H$. Using  the notation of the second quantization it can be written as
\[H:=\mathop\oplus\limits_{N=0}^\infty H_N=\sum_\p \p^2a_\p^\dagger a_\p+\frac{1}{2N}
\sum_{\p,\q,\k}\hat v(\k)a_{\p+\k}^\dagger a_{\q-\k}^\dagger a_\q a_\p.\]

If $A$ is an operator on the one-particle space, then by its {\em second quantization} we will mean the operator that on the $N$-particle space equals
\[\sum_{i=1}^N A_i.\]
If we use an orthonormal basis,
 say, $|\p\rangle$, $\p\in\frac{2\pi}{L}\zz^d$, then
 this operator written in the 2nd quantized notation equals
\[\frac12\sum_{\p_1,\p_2 }\langle \p_1|A|\p_2\rangle
a_{\p_1}^\dagger
a_{\p_2}.
\]

Let us introduce some special notation for various operators and their 2nd quantization. 

Let $P$ be the projection onto the constant function in $L^{2}(]L/2,L/2]^d)$, and $Q=1-P$. The operator that counts the number of particles in, resp. outside the zero momentum mode will be denoted by $N_0$, resp. $N^{>}$, i.e.
\begin{gather}
N_0=\sum_{i=1}^{N}P_{i},\ \ \ 
N^{>}=\sum_{i=1}^{N}Q_{i}. \label{particlenoop}
\end{gather}
In the 2nd quantization notation,
\[N_0=a_\0^\dagger a_\0,\ \ \ N^>=\sum_{\p\neq\0}a_\p^\dagger a_\p.\]

For $N$-particle bosonic wave functions $\Psi,\Phi$ we have
\begin{eqnarray}
\langle\Psi|N^{>}|\Phi\rangle&=&N\langle\Psi|Q_{1}|\Phi\rangle,\label{Q1}\\ 
\langle\Psi|N^{>}(N^>-1)|\Phi\rangle&=&N(N-1)\langle\Psi|Q_{1}Q_2|\Phi\rangle.  \label{Q1Q2}
\end{eqnarray}

The symbol $T$ will denote the kinetic energy of the system: $T=-\sum\limits_{i=1}^{N}\Delta_{i}$. For further reference, note that
\beqn
\langle\Psi|N^{>}|\Psi\rangle \leq \frac{L^{2}}{(2\pi)^{2}}\langle\Psi|T|\Psi\rangle
. \label{lemma1b-}  
\eeqn

We will also need the notion of
 the second quantization of certain 2-body operators. More precisely, let $w$ be an operator on the symmetrized 2-particle space. Then by its second quantization we will mean the operator that restricted to the 
$N$-particle space equals
\[\sum_{1\leq i<j\leq N} w_{ij}.\]
If $w$ is an operator on the unsymmetrized 2-particle space, then we can also speak about its second quantization, but now its restriction to
the 
$N$-particle space equals
\[\frac12\sum_{1\leq i\neq j\leq N} w_{ij}.\]
In the momentum
 basis  this operator written in the 2nd quantized language equals
\[\frac12\sum_{\p_1,\p_2,\p_3,\p_4}\langle \p_1,\p_2|w|\p_3,\p_4\rangle
a_{\p_1}^\dagger
a_{\p_2}^\dagger
a_{\p_3}
a_{\p_4}.
\]

\section{Bounds on interaction}

The potential $v$ can be interpreted as an operator of multiplication
by $v(\x_1-\x_2)$ on $L_\s^2(\Lambda^2)$.
Following \cite{S}, we would like to estimate this 2-body operator by simpler, 1-body operators. As a preliminary step we record the following bound:

\begin{lemma} Let $\epsilon>0$. Then
\begin{eqnarray*}
v&\geq& P\otimes PvP\otimes P+P\otimes PvQ\otimes Q+Q\otimes QvP\otimes P \\&&+(1-\epsilon)(P\otimes Q+Q\otimes P)v(P\otimes Q+Q\otimes P)\\&&+(1-\epsilon^{-1})Q\otimes QvQ\otimes Q,\\
v&\leq& P\otimes PvP\otimes P+P\otimes PvQ\otimes Q+Q\otimes QvP\otimes P \\&&+(1+\epsilon)(P\otimes Q+Q\otimes P)v(P\otimes Q+Q\otimes P)\\&&+(1+\epsilon^{-1})Q\otimes QvQ\otimes Q.
\end{eqnarray*} \label{rhs}
\end{lemma}
\begin{proof} Using the translation invariance of $v$ we obtain
\begin{eqnarray*}
v&=& (P\otimes P+Q\otimes Q)v (P\otimes P+Q\otimes Q)
 \\&&+(P\otimes Q+Q\otimes P)v(P\otimes Q+Q\otimes P)
\\&&+(P\otimes Q+Q\otimes P)vQ\otimes Q
+Q\otimes Qv(P\otimes Q+Q\otimes P).
\end{eqnarray*} Then we apply the Schwarz inequality to the last two terms.
\end{proof}

Let us now identify the second quantization of various terms on the r.h.s of the estimates of Lemma \ref{rhs}:
\begin{eqnarray*}
 P\otimes PvP\otimes P&&
\frac{1}{2L^d}\hat v(\0)a_\0^\dagger a_\0^\dagger a_\0a_\0=
\frac{1}{2L^d}\hat v(\0)N_0(N_0-1),\\
P\otimes PvQ\otimes Q&&
\frac{1}{2L^d}\sum_{\p\neq\0} \hat v(\p)a_0^\dagger a_0^\dagger a_\p a_{-\p},\\
Q\otimes QvP\otimes P&&
\frac{1}{2L^d}\sum_{\p\neq\0} \hat v(\p)a_\p^\dagger a_{-\p}^\dagger a_0 a_0,\\
P\otimes QvQ\otimes P,\ Q\otimes PvP\otimes Q
&&
\frac{1}{2L^d}\sum_{\p\neq\0} \hat v(\0)a_\p^\dagger  a_{\p}N_0,\\
P\otimes QvP\otimes Q,\ Q\otimes PvQ\otimes P
&&
\frac{1}{2L^d}\sum_{\p\neq\0} \hat v(\p)a_\p^\dagger  a_{\p}N_0
.
\end{eqnarray*} 
The second quantization of $ Q\otimes QvQ\otimes Q$
 can be bounded  from above by
\[v(\0)\sum_{1\leq i<j\leq  N}Q_iQ_j=v(\0)\frac12N^>(N^>{-}1).\]

Introduce the family of {\em estimating Hamiltonians}
\begin{eqnarray*}
H_{N,\epsilon}&:= &\frac{1}{2}\hat v(\0)(N-1)+\sum_{\p\neq\0}\big(|\p|^2+\hat v(\p)\big)
a_\p^\dagger a_\p\\
&&+\frac{1}{2N}\sum_{\p\neq\0}\hat v(\p)\Big(a_\0^\dagger a_\0^\dagger a_\p a_{-\p}+a_\p^\dagger a_{-\p}^\dagger a_\0 a_\0 \Big)\\
&&-\frac{1}{N}\sum_{\p\neq\0}\big(\hat v(\p)+\frac{\hat v(\0)}{2}\big)a_\p^\dagger a_\p N^>+\frac{\hat v(\0)}{2N}N^>\\
&&+\frac{\epsilon}{N}\sum_{\p\neq\0}\big(\hat v(\p)+\hat v(\0)\big)a_\p^\dagger a_\p N_0\\
&&+(1+\epsilon^{-1})\frac{1}{2N} v(\0)L^d N^>(N^>-1)
\end{eqnarray*}

The operators $H_{N,\epsilon}$ preserve the $N$-particle sectors. By the above calculations  we obtain the following estimates on the Hamiltonian:
\begin{eqnarray}
H_N&\geq &H_{N,-\epsilon}, \ \  0<\epsilon\leq1;\label{hamiltonianboundlower}  \\
H_N&\leq &H_{N,\epsilon},\ \ 0<\epsilon. \label{hamiltonianboundupper}
\end{eqnarray}
\section{Extended space}
So far we used the {\em physical  Hilbert space} (\ref{physicalspace1}). By the exponential property of Fock spaces we have the identification
\begin{equation}
\mathcal{H}\simeq
\Gamma_\s(\cc)\otimes\Gamma_\s\Big(l^2\bigl(\frac{2\pi}{L}\zz^d\backslash\{\0\}\bigr)\Bigr).
\label{physicalspace}
\end{equation}
Let us embed the space of  {\em zero modes} $\Gamma_\s(\cc)=l^2(\{0,1,\dots\})$ in a larger space $l^2(\zz)$. Thus we obtain the {\em extended Hilbert space}
\begin{equation}
\mathcal{H}^{\text{ext}}:= l^2(\zz)\otimes \Gamma_\s\Big(l^2\bigl(\frac{2\pi}{L}\zz^d\backslash\{\0\}\bigr)\Bigr).
\label{fullspace}
\end{equation} 
The physical space (\ref{physicalspace}) is spanned by vectors of the
form $|n_0\rangle\otimes \Psi^>$, where
 $|n_0\rangle$ represents $n_0$  zero modes  ($n_0\geq0$)
and $ \Psi^>$ represents a vector outside the zero mode.

The space (\ref{fullspace}) is also spanned by vectors  of this  form, where now 
 the relation $n_0\geq0$ is not imposed. 
The orthogonal complement of $\cH$ in $\cH^\ext$ will be denoted by $\cH^\nph$ (for ``non-physical''). 

 On $\cH^\ext$ we have a self-adjoint operator $N_0^\ext$ such that $N_0^\ext|n_0\rangle\otimes \Psi^>=n_0|n_0\rangle\otimes \Psi^>$. Its spectrum equals $\zz$. 
 Clearly
\[N_0^\ext\Big|_\cH=N_0,\ \ \ \cH=\Ran \one_{[0,\infty[}(N_0^\ext),\  \ 
\cH^\nph=\Ran \one_{]-\infty,0[}(N_0^\ext).\]
If $N\in\zz$, we will write $\cH_N^\ext$ for the subspace of $\cH^\ext$ corresponding to $N^>+N_0^\ext=N$.

We have also a unitary operator
\begin{gather*}
U|n_0\rangle\otimes\Psi^>=|n_0-1\rangle\otimes\Psi^>.
\intertext{Notice that both $U$ and $U^{\dagger}$ commute with both $a_{\p}$ and $a^{\dagger}_{\p}$ with $\p\neq\0$.  We now define for $\p\neq\0$ the following operator on $\mathcal{H}^{\text{ext}}$:}
{b}_{\p}:=a_{\p}U^{\dagger}. 
\end{gather*}
Operators $b_\p$ and $b_\q^\dagger$ satisfy the same CCR as $a_\p$ and $a_\q^\dagger$.

The extended space is useful in the study of $N$-body Hamiltonians.
To illustrate this, on $\cH_N^\ext$ let
us introduce the {\em extended Hamiltonian}
\begin{eqnarray*}
H_N^\ext&=&\sum_{\p\neq\0}\Big(\p^2+\frac{N_0^\ext}{N}(\hat v(\p)+\hat
v(\0)\Big)b_\p^\dagger
b_\p\\&&
+
\frac{1}{2}\sum_{\p\neq\0}\Big(\hat v(\p)\frac{\sqrt{N_0^\ext(N_0^\ext
    -1)}}{N}
b_\p
b_{-\p}
+\hc\Big)\\
&&+\frac{1}{N}\sum_{\k,\p\neq\0}\hat v(\k)b_\k^\dagger b_{\p-\k}^\dagger
b_\p\sqrt{\max(N_0^\ext,0)}+\hc\Big)\\
&&+\frac{1}{2N}
\sum_{\p,\q,\k\neq\0}\hat v(\k)b_{\p+\k}^\dagger b_{\q-\k}^\dagger
b_\q b_\p.
\end{eqnarray*}
It is easy to see that $H_N^\ext$ preserves the $N$-particle physical
space $\cH_N$ and on $\cH_N$  it coincides with $H_N$.

In our paper we will use the {\em extended estimating Hamiltonian}, which is
the following operator on $\cH_N^\ext$:
\begin{eqnarray*}
H_{N,\epsilon}^\ext&:= &\frac{1}{2}\hat v(\0)(N-1)+\sum_{\p\neq\0}\big(|\p|^2+\hat v(\p)\big)
b_\p^\dagger b_\p\\
&&+\frac{1}{2}\sum_{\p\neq\0}\hat v(\p)\Big(\frac{\sqrt{(N_0^\ext-1)N_0^\ext}}{N}b_\p b_{-\p}+\hc
\Big)\\&&-\frac{1}{N}\sum_{\p\neq\0}\big(\hat v(\p)+\frac{\hat v(\0)}{2}\big)b_\p^\dagger b_\p N^>+\frac{\hat v(\0)}{2N}N^>\\
&&+\frac{\epsilon}{N}\sum_{\p\neq\0}\big(\hat v(\p)+\hat v(\0)\big)b_\p^\dagger b_\p N_0^\ext\\
&&+(1+\epsilon^{-1})\frac{1}{2N} v(\0)L^d N^>(N^>-1).
\end{eqnarray*}
Note that $H_{N,\epsilon}^\ext$ 
preserves $\cH_N$ and restricted to $\cH_N$ coincides with $H_{N,\epsilon}$.

\section{Bogoliubov Hamiltonian}

Consider the  operator
\begin{eqnarray*}
&&\sum_{\p\neq\0}\big(|\p|^2+\hat v(\p)\big)
b_\p^\dagger b_\p+\frac{1}{2}\sum_{\p\neq\0}\hat v(\p)\Big(b_\p b_{-\p}+b_\p^\dagger b_{-\p}^\dagger\Big).
\end{eqnarray*} 
acting on $\cH^\ext$.
It commutes with $N_0+N^>$ and $U$. 
In particular, it preserves
$\cH_N^\ext$. Its restriction to $\cH_N^\ext$ will be denoted
$H_{\Bog,N}$.

We can write
\begin{eqnarray}
H_{N,\epsilon}^\ext&=&\frac{1}{2}\hat v(\0)(N-1)+H_{\Bog,N}+R_{N,\epsilon}, \label{hamextepsilon}\\
R_{N,\epsilon}&:=
&\frac{1}{2}\sum_{\p\neq\0}\hat v(\p)\Big(\Big(\frac{\sqrt{(N_0^\ext-1)N_0^\ext}}{N}
-1\Big)b_\p b_{-\p}
+\hc\Big) \notag\\&&-\frac{1}{N}\sum_{\p\neq\0}\big(\hat v(\p)+\frac{\hat v(\0)}{2}\big)b_\p^\dagger \b_\p N^>+\frac{\hat v(\0)}{2N}N^>\notag\\
&&+\frac{\epsilon}{N}\sum_{\p\neq\0}\big(\hat v(\p)+\hat v(\0)\big)b_\p^\dagger b_\p N_0^\ext \notag\\
&&+(1+\epsilon^{-1})\frac{1}{2N} v(\0)L^d N^>(N^>-1). \label{rnepsilon}
\end{eqnarray}

Clearly, all $H_{\Bog,N}$ are unitarily equivalent to one another:
$UH_{\Bog,N} U^\dagger=H_{\Bog,N-1}$. It is easy to see that they are
all unitarily equivalent to what we can call the {\em standard Bogoliubov Hamiltonian}:
\begin{eqnarray}
H_\Bog&=&\sum_{\p\neq\0}\big(|\p|^2+\hat v(\p)\big)
a_\p^\dagger a_\p+\frac{1}{2}\sum_{\p\neq\0}\hat v(\p)\Big(a_\p a_{-\p}+a_\p^\dagger a_{-\p}^\dagger\Big).\label{stan}
\end{eqnarray}
$H_\Bog$ acts on $\Gamma_\s\Big(l^2\bigl(\frac{2\pi}{L}\zz^d\backslash\{\0\}\bigr)\Bigr)$.

We would now like to find a unitary transformation diagonalizing $H_\Bog$. To this end set
\begin{equation*}
A_{\p}:=|\p|^{2}+\hat{v}(\p), \,\,\,\,\,\,B_{\p}:= \hat{v}(\p).
\end{equation*}
Introduce also $\alpha_\p$, $\beta_\p$, $c_\p$ and $s_\p$ by
\begin{eqnarray*}
\alpha_{\p}&=&\frac{1}{B_\p}\Big(A_\p-\sqrt{A_\p^2-B_\p^2}\Big)\ =\ \tanh(2\beta_\p),\\
c_\p&=&\frac{1}{\sqrt{1-\alpha_\p^2}}\ =\ \cosh(2\beta_\p),\\
s_\p&=&\frac{\alpha_\p}{\sqrt{1-\alpha_\p^2}}\ =\ \sinh(2\beta_\p).
\end{eqnarray*}
Now let $S=\e^{-X}$,
where
\begin{gather}
X=\sum_{\p\neq\0}\beta_{\p}\left(a_{\p}^{\dagger}a_{-\p}^{\dagger}-a_{\p}a_{-\p}\right). \label{X} 
\end{gather}
 Then using the Lie formula
\begin{gather*}
\e^{-X}a_{\q}\e^X=\sum_{j=0}^\infty\frac{(-1)^j}{j!}\underset{\hbox{ $j$ times\hskip 3ex}}{[X,...[X,a_{\q}]\dots]} \\
=1+2\beta_{\q}a^{\dagger}_{-\q}+\frac12 4\beta^{2}_{\q}a_{\q}+\ldots
\end{gather*} 
we get 
\begin{gather}
Sa_{\q}S^{\dagger}=c_\q a_{\q}+s_\q a^{\dagger}_{-\q}. \label{sas}
\end{gather}
Therefore,
\begin{eqnarray}
H_{\Bog}&=&\sum_{\p\neq 0}\frac{1}{2}\left(A_{\p}(a_{\p}^{\dagger}a_{\p}+a_{-\p}^{\dagger}a_{-\p})+B_{\p}(a_{\p}^{\dagger}a_{-\p}^{\dagger}+a_{\p}a_{-\p})\right) \notag \\
&=&-\frac{1}{2}\sum_{\p\neq 0}\left(A_{\p}-\sqrt{A_{\p}^{2}-B_{\p}^{2}}\right)\\&&+
\sum_{\p\neq 0}\sqrt{A_{\p}^{2}-B_{\p}^{2}}
\big(c_\p a_\p^\dagger+s_\p a_{-\p}\big)\big(c_\p a_\p+s_\p a_{-\p}^\dagger\big)\\
&=&E_\Bog+S\big(\sum_{\p\neq 0}e_{\p}a_{\p}^{\dagger}a_{\p}\big)S^\dagger,
\label{hbogoliubov}
\end{eqnarray}
where $e_{\p}$ and $E_{\Bog}$ are defined in  the introduction.
Thus the spectrum of $H_\Bog-E_\Bog$  equals 
\[
\Big\{\sum_{i=1}^j e_{\k_i}\ :\ 
\k_1,\dots,\k_j\in\frac{2\pi}{L}\zz^d\backslash\{\0\},\ \ 
j=0,1,2,\dots\Big\}.\]

For further reference note the identities:
\begin{eqnarray}
\alpha_\p&=&\frac{\hat v(\p)}{|\p|^2+\hat v(\p)+
|\p|\sqrt{2\hat v(\p)+|\p|^2}},\notag\\
(c_\p-s_\p)^2=\frac{1-\alpha_\p}{1+\alpha_\p}&=&
\frac{|\p|}
{\sqrt{|\p|^2+2\hat v(\p)}}, \label{(cp-sp)2} \\
s_\p(c_\p-s_\p)=\frac{\alpha_\p}{1+\alpha_\p}&=&
\frac{\hat v(\p)}
{|\p|^2+2\hat v(\p)+|\p|\sqrt{|\p|^2+2\hat v(\p)}}, \label{s(c-s)}\\
2s_\p c_\p(c_\p-s_\p)^2=\frac{\alpha_\p}{(1+\alpha_\p)^2}&=&
\frac{\hat v(\p)}{|\p|^2+2\hat v(\p)}.\notag
\end{eqnarray}
We note also an alternative formula for the Bogoliubov energy:
\[E_\Bog
=
-\frac{1}{2}\sum_{\p\in\frac{2\pi}{L}\mathbb{Z}^{d}\setminus \{ 0\}}\frac{\hat v(\p)^2}{|\p|^{2}+\hat{v}(\p)+|\p|\sqrt{|\p|^2+2\hat{v}(\p)|\p|}}. 
\]

\section{Lower bound}

In this section we prove the lower bound part of Thm \ref{thm}. Using the notation introduced in the previous sections it follows from the following statement:

\begin{thm} Let $c>0$. Then  there exists $C$ such that for any $\kappa\geq 0$ with
 \beq L^{d+2}(L^d+\kappa)\leq cN\label{condit1}\eeq
 we have
\begin{eqnarray*}
\vecsp\big(\one_{[0,\kappa]}(H_N-E_N)H_N\big)&\geq&
\frac{1}{2}\hat{v}(\0)(N-1)+\vecsp\big(H_\Bog\big)\\&&-CN^{-1/2}L^{d/2+3}(\kappa+L^d)^{3/2}.
\end{eqnarray*}
\label{lower}\end{thm}

The proof of the  lower bound starts with estimates analogous to Lemmas 1 and 2 of \cite{S}. Note that in these estimates all operators involve the physical  Hilbert space.

\begin{lemma}\label{lemma1}
The ground state energy $E_N$ of $H_N$ satisfies the bounds
\beqn
 0\geq E_N - \frac{1}{2}\left(N-1\right)\hat{v}(\0)\geq \frac{1 }{2}\big(\hat v(\0)-L^d v(\0)\big). \label{lemma1a}
\eeqn
 \label{lemma1b}  
\end{lemma}
\begin{proof}
The upper bound to the ground state energy follows by using a constant trial wave function $\Psi=L^{-Nd/2}$, which gives
\beqn\frac{1}{2}(N-1)\hat v(\0)\geq E_N.\label{lemma1acalc2}\eeqn

Using $\hat{v}(\p)\geq 0$ for every $\p\in\frac{2\pi}{L}\Z^{d}$ we obtain
 $\sup\limits_\x v(\x)=v(\0)$. Moreover, \beqnn
\frac{1}{2L^d}\sum_{\p\in\frac{2\pi}{L}\Z^{d}\setminus\{\0\}}\hat{v}(\p)\left|\sum_{i=1}^{N}\e^{\i\p\x_{j}}\right|^{2}\geq0.
\eeqnn
This is equivalent to
\beqn
\sum_{1\leq i<j\leq N}v(\x_{i}-\x_{j})\geq\frac{N^{2}}{2L^d}\hat{v}(\0)-\frac{N}{2}v(\0). \label{lemma1acalc}
\eeqn
Hence,
\beqn
H_N\geq T + \frac{ L^d}{N}  \left(\frac{N^{2}}{2L^d}\hat{v}(\0)-\frac{N}{2}v(\0)\right),
\label{lemma1acalc1}\eeqn
and so
\[E_N\geq \frac{ L^d}{N}   \left(\frac{N^{2}}{2L^d}\hat{v}(\0)-\frac{N}{2}v(\0)\right).\]
\end{proof} 

  Let $\kappa\geq0$. 
For brevity we  introduce the following notation for  the spectral
projection onto  the spectral subspace of $H_N$ corresponding to the
energy less than or equal to $E_N+\kappa$:
\[\one_{\kappa}^N:=\one_{[0,\kappa]}(H_N-E_N).\]
 $\one_{\kappa}^N$ can be  understood as a projection acting on the extended space with range in the physical space.

\begin{lemma} There exists $C$ such that
\beq N^>\leq CL^{2}(H_N-E_N+L^d).\label{pafa}\eeq
Consequently,
\begin{eqnarray}
\label{eq6a}
\one_{\kappa}^NN^{>}\one_{\kappa}^N&\leq &C L^{2}\left(L^d+\kappa\right).
\end{eqnarray}
\end{lemma}

\begin{proof}
Using first (\ref{lemma1acalc1}) and   
(\ref{lemma1acalc2}) we obtain
\begin{eqnarray*}
T&\leq &H_N-E_N-\frac{1}{2}\hat v(\0)+\frac{1}{2}L^d v(\0)\\
&\leq& C(H_N-E_N+L^d).
\end{eqnarray*}
By  (\ref{lemma1b-}) this implies (\ref{pafa}).
\end{proof}

\begin{lemma}\label{lemma2}
We have
\begin{eqnarray}
\label{eq6}
\one_{\kappa}^N(N^{>})^2\one_{\kappa}^N&\leq &C L^{4}\left(L^d+\kappa\right)^{2}.
\end{eqnarray}
\end{lemma}
\begin{proof}
Let $\one_{\kappa}^N\Psi=\Psi$.
As in \cite{S},
\begin{eqnarray}
\langle\Psi|N^{>}T|\Psi\rangle&=&\langle\Psi|N^{>}(H_N-E_N-\frac{1}{2}\kappa)|\Psi\rangle\label{eq1}\\
&&+
N\Big\langle \Psi|Q_1\Big(E_N+\frac12\kappa-\frac{L^d}{N}\sum_{2\leq i<j\leq N}v(\x_{i}-\x_{j})\Big)|\Psi\Big\rangle\label{eq2}\\
&&
-L^d\Big\langle \Psi|Q_1\sum_{2\leq j\leq N}v(\x_{1}-\x_{j})|\Psi\Big\rangle\label{eq3}.
\end{eqnarray}
Using  Schwarz's inequality, the first term can be bounded as
\begin{eqnarray*}
|(\ref{eq1})|&\leq&\|N^>\Psi\|\,\Big\|H_N-E_N-\frac12\kappa\Big\|\\
&\leq&\frac\kappa2\langle\Psi|(N^>)^2\Psi\rangle^{1/2}.
\end{eqnarray*}
Let us estimate the second term. Using (\ref{lemma1acalc}) we get
\begin{eqnarray*}
E_N-\frac{L^d}{N}\sum_{2\leq i<j\leq N}v(\x_{i}-\x_{j})&\leq&\frac{1}{2}(N-1)\hat v(\0)+\frac{L^d}{2N}(N-1)v(\0)-\frac{1}{2N}(N-1)^2\hat v(\0)\\
&=&\frac{1}{2}\frac{N-1}{N}\big(\hat v(\0)+L^d v(\0)\big).
\end{eqnarray*}
Hence,
\begin{eqnarray*}
(\ref{eq2})&\leq&
\Big(\frac\kappa2+\frac12\frac{N-1}{N}\big(\hat v(\0)+L^dv(\0)\big)\Big)N\langle\Psi|Q_1|\Psi\rangle\\
&\leq &
\Big(\frac\kappa2+\frac12\big(\hat v(\0)+L^dv(\0)\big)\Big)\langle\Psi|N^>|\Psi\rangle
\end{eqnarray*}
Finally, let us consider the third term:
\begin{eqnarray*}
\langle\Psi|Q_1 v(\x_1-\x_2)|\Psi\rangle&=&
\langle \Psi|Q_1Q_2 v(\x_1-\x_2)|\Psi\rangle+
\langle \Psi|Q_1P_2 v(\x_1-\x_2)Q_2|\Psi\rangle\\
&&+\langle \Psi|Q_1P_2 v(\x_1-\x_2)Q_2|\Psi\rangle
,\\[3ex]
|\langle\Psi|Q_1 Q_2v(\x_1-\x_2)|\Psi\rangle|&\leq&
v(\0)\langle \Psi|Q_1Q_2|\Psi\rangle^{1/2},\\
|\langle\Psi|Q_1 P_2v(\x_1-\x_2)Q_2|\Psi\rangle|&\leq&
v(\0)\langle \Psi|Q_1|\Psi\rangle,\\
\langle\Psi|Q_1 P_2v(\x_1-\x_2)P_2|\Psi\rangle|&=&
\hat v(\0)\langle \Psi|Q_1P_2|\Psi\rangle\geq0.
\end{eqnarray*}
Therefore, using (\ref{Q1}) and (\ref{Q1Q2})
\begin{eqnarray*}
|(\ref{eq3})|&\leq&
v(\0)L^d\Big(\sqrt{\frac{N-1}{N}}\langle\Psi|(N^>{-}1)N^>|\Psi\rangle^{1/2}
+\frac{N-1}{N}\langle\Psi|N^>|\Psi\rangle\Big)\\&\leq&
v(\0)L^d
\Big(\langle\Psi|(N^>)^2\Psi\rangle^{1/2}+
\langle\Psi|N^>|\Psi\rangle\Big).
\end{eqnarray*}
Now 
\begin{eqnarray}
\langle \psi|(N^{>})^{2}|\psi\rangle
&\leq & \frac{L^{2}}{(2\pi)^{2}}\langle \psi|N^{>}T|\psi\rangle.\label{eq4}
\end{eqnarray}
We can add the three estimates, use (\ref{eq4}) and obtain
\begin{eqnarray*}
\langle\Psi|N^>T|\Psi\rangle&
\leq&C(\kappa +L^d)\Big(\langle\Psi|(N^>)^2|\Psi\rangle^{1/2}+
\langle\Psi|N^>|\Psi\rangle\Big)\\
&\leq &CL^{2}(\kappa +L^d)^{2}\\&&+
CL(\kappa +L^d)
\langle\Psi|N^>T|\Psi\rangle^{1/2}.
\end{eqnarray*}
Setting
 $X:=\langle \psi|N^{>}T|\psi\rangle^{1/2}$ we can rewrite this as $X^{2}<c+aX$
in the obvious notation. Solving this inequality we get that
\[X^{2}\leq\frac{a^{2}}{2}+c+\sqrt{a^{2}+4c}.\]
This implies
\begin{eqnarray}\label{eq5}
\one_{\kappa}^NN^{>}T\one_{\kappa}^N&\leq &CL^{2}\left(L^d+\kappa\right)^{2}.\end{eqnarray}
If in addition we use  (\ref{eq4}), we obtain
 (\ref{eq6}).
\end{proof}

\begin{lemma} 
\beq
\sup_{0<\epsilon\leq1}\one_{\kappa}^NR_{N,-\epsilon}\one_{\kappa}^N
\geq -CN^{-1/2}L^{d/2+3}(L^d+\kappa)^{3/2}.\eeq
\label{lami}\end{lemma}

\begin{proof}
\begin{eqnarray}\one_{\kappa}^NR_{N,-\epsilon}\one_{\kappa}^N
&\geq&\one_{\kappa}^N\frac{1}{2}\sum_{\p\neq\0}\hat v(\p)\Big(\Big(\frac{\sqrt{(N_0-1)N_0}}{N}
-1\Big)b_\p b_{-\p}+\hc
\Big)\one_{\kappa}^N \notag\\
&&-\one_{\kappa}^N\frac{1}{N}\sum_{\p\neq\0}\big(\hat v(\p)+\frac{\hat v(\0)}{2}\big)b_\p^\dagger b_\p N^>\one_{\kappa}^N \notag\\
&&-\epsilon\one_{\kappa}^N\frac{1}{N}\sum_{\p\neq\0}\big(\hat v(\p)+\hat v(\0)\big)b_\p^\dagger b_\p N_0\one_{\kappa}^N\notag\\
&&-\epsilon^{-1}\one_{\kappa}^N\frac{1}{2N} v(\0)L^d (N^>)^2 \one_{\kappa}^N.
\label{pada} 
\end{eqnarray}
Note that the range of  $\one_{\kappa}^N$ is inside the physical space, so whenever possible we replaced $N_0^\ext$ by $N_0$.
It is easy to estimate from below
 various terms on the right of (\ref{pada}) by expressions involving $N^>$. 
The first term requires more work than the others. We have
\begin{eqnarray*}
N-\sqrt{(N_0-1)N_0}&=&\frac{2NN^>-(N^>)^2+N-N^>}
{N+\sqrt{(N-N^>-1)(N-N^>)}}\\
&\leq&2N^>+1.
\end{eqnarray*}
Then we use
\begin{eqnarray*}
\big(\sqrt{(N_0-1)N_0}-N\big)\sum_{\p\neq\0}\hat v(\p)b_\p b_{-\p}+\hc&\geq
&-\Big(\sum_{\p\neq\0}\hat v(\p)b_\p b_{-\p}\Big)^\dagger
\sum_{\p\neq\0}\hat v(\p)b_\p b_{-\p}\\
&&-\big(\sqrt{(N_0-1)N_0}-N\big)^2\\
&\geq&-C(N^>)^2-(2N^>+1)^2\\& \geq&
 -C_1\big((N^>)^2+1\big).
\end{eqnarray*}

To bound the third term 
we use   $N_0\leq N$. We obtain 
\begin{eqnarray*}
\one_{\kappa}^NR_{N,-\epsilon}\one_{\kappa}^N&\geq&-C\one_{\kappa}^N\frac{(N^>)^2+1}{N}\one_{\kappa}^N\\
&&-C\one_{\kappa}^N\frac{(N^>)^2}{N}\one_{\kappa}^N\\
&&-\epsilon C\one_{\kappa}^N N^>\one_{\kappa}^N\\
&&-\epsilon^{-1}C\one_{\kappa}^NL^d\frac{(N^>)^2}{N}\one_{\kappa}^N.
\end{eqnarray*}
Using that $0\leq \epsilon\leq1$ and $L\geq1$, we can partly absorb the first two terms in the fourth:
\begin{eqnarray*}
&\geq&-\frac{C}{N}\one_{\kappa}^N-\epsilon C\one_{\kappa}^N N^>\one_{\kappa}^N-\epsilon^{-1}C\one_{\kappa}^NL^d\frac{(N^>)^2}{N}\one_{\kappa}^N.
\end{eqnarray*}
By (\ref{pafa}) and (\ref{eq6}), this can be estimated by
 \begin{eqnarray}
&\geq&-CN^{-1} -\epsilon CL^{2}(L^d+\kappa)-
\epsilon^{-1}CN^{-1}L^{d+4}(L^d+\kappa)^2.
\label{errorlowerbound}
\end{eqnarray}
 Setting $\epsilon= c^{-1/2}L^{d/2+1}(L^d+\kappa)^{1/2}N^{-1/2}$ in (\ref{errorlowerbound}), which by  Condition (\ref{condit1}) is less than $1$,  
we bound it by
\[\geq\ -CN^{-1}-CN^{-1/2}L^{d/2+3}(L^d+\kappa)^{3/2}.\]
Using $L\geq1$, we can absorb the first term in the second. \end{proof}

\begin{proof}[Proof of Thm \ref{lower}]
 Recall inequality (\ref{hamiltonianboundlower}), which
implies for $0<\epsilon\leq1$
\beq
\one_{\kappa}^N H_N\one_{\kappa}^N\geq
\one_{\kappa}^N
\left(\frac{1}{2}\hat v(\0)(N-1)+H_{\Bog,N}+R_{N,-\epsilon}\right)\one_{\kappa}^N.
\eeq
Thus it suffices to apply Lemma \ref{lami} and  the min-max principle.
\end{proof}

\begin{proof}[Proof of Thm \ref{thm} (1)]
First set $\kappa=0$. Then Condition (\ref{condit1}) becomes Condition (\ref{war1a}) and we obtain Thm \ref{thm} (1a).

Next set $\kappa=K_N^j(\p)$. Then Condition (\ref{condit1}) is equivalent to the conjunction of Conditions (\ref{war1a}) and (\ref{war1b}). We obtain Thm \ref{thm} (1b).
\end{proof}

\section{Upper bound}

In this section we prove the following theorem, which implies the upper bound of Thm \ref{thm}:

\begin{thm}
Let $c>0$. Then there exist $c_1>0$ and $C$ such that if $\kappa\geq0$ and \begin{eqnarray}
L^{d+2}(\kappa+L^{d-1})&\leq &cN,\label{condit2}\\
L^{2}(\kappa+L^{d-1})&\leq &c_1N\label{condit2a}
\end{eqnarray}
then 
\begin{eqnarray*}
\vecsp\big(H_N\big)&\leq&
\frac{1}{2}\hat{v}(\0)(N-1)+\vecsp\Big(\one_{[0,\kappa]}(H_\Bog-E_\Bog)H_\Bog
\Big)
\\&&+CN^{-1/2}L^{d/2+3}(\kappa+L^{d-1})^{3/2}.
\end{eqnarray*}
\label{upper}\end{thm}

For brevity, we set
\[\one_\kappa^\Bog:=\one_{[0,\kappa]}(H_{\Bog,N}-E_\Bog).\]
 From now on, to simplify the notation we will also write $H_\Bog$ instead of $H_{\Bog,N}$, even though this is an abuse of notation. ($H_{\Bog,N}$ is  unitarily equivalent, but strictly speaking distinct from (\ref{stan})).

We also set
\[d_\p:=Sb_\p S^\dagger\]
 where $S$ is defined as in \eqref{X}  with operators $a$'s replaced by $b$'s.
Clearly,
\begin{eqnarray*}
d_\p&=\ c_\p b_{\p}+s_\p b_{-\p}^\dagger,\ \ \ d_\p^\dagger&=\ c_\p b_\p^\dagger+s_{\p}b_{-\p}.
\end{eqnarray*}

\begin{lemma}
There exist $C_1,C_2$ such that
\beq H_\Bog-E_\Bog\geq C_1L^{-2}N^>-C_2L^{d-1}.\label{qe1}\eeq
Consequently,
\beq
\one_\kappa^\Bog  N^>\one_\kappa^\Bog\leq CL^{2}(L^{d-1}+\kappa).
\label{qe1a}\eeq
\label{lemi1}\end{lemma}
\proof
Using (\ref{hbogoliubov})
 we have that 
\begin{eqnarray*}H_{\Bog}-E_\Bog&=&\sum_{\p\neq 0}e_{\p}Sb_{\p}^{\dagger}b_{\p}S^{\dagger}\\
&\geq&\sum_{\p\neq 0} \frac{\pi\sqrt{8\hat v(\0)}}{L}Sb_{\p}^{\dagger}b_{\p}S^\dagger=
\frac{\pi\sqrt{8\hat v(\0)}}{L}SN^{>}S^{\dagger}.\end{eqnarray*}
Now
\begin{eqnarray*}SN^{>}S^{\dagger}\ =\ \sum_{\pm\p\neq\0} \big(d_\p^\dagger d_\p+d_{-\p}^\dagger d_{-\p}\big)&=&
\sum_{\pm\p\neq\0}\Big( (c_\p^2+s_\p^2)\big(b_{\p}^\dagger b_{\p}+b_{-\p}^\dagger b_{-\p}\big)\notag\\
&&+ 2c_\p s_\p\big(b_{\p}^\dagger b_{-\p}^\dagger+b_{\p} b_{-\p}\big)+ 2s_\p^2\Big).
\end{eqnarray*}
(When we write $\pm \p$ under the summation symbol, we sum over all pairs $\{\p,-\p\}$). Using
\[b_{\p}^\dagger b_{-\p}^\dagger+b_{\p} b_{-\p}\geq
-\big(b_{\p}^\dagger b_{\p}+b_{-\p}^\dagger b_{-\p}+1\big)
\]we obtain
\begin{eqnarray} &&\sum_{\pm\p\neq\0}\big(d_\p^\dagger d_\p+d_{-\p}^\dagger d_{-\p}\big)\notag\\&\geq&\sum_{\pm\p\neq\0}\Big(
(c_\p-s_\p)^2\big(b_\p^\dagger b_\p+b_{-\p}^\dagger b_{-\p}\big)-2s_\p(c_\p-s_\p)\Big).\label{sns}
\end{eqnarray}
By \eqref{(cp-sp)2} we know that $\inf\limits_{\p\neq 0}(c_{\p}-s_{\p})^2\geq\frac{\sqrt{2}\pi}{\sqrt{\hat{v}(0)}L}$. Also, \eqref{s(c-s)} yields
\[
\frac{1}{L^d}\sum_{\pm\p\neq\0}s_{\p}(c_{\p}-s_{\p})<\infty,\]
uniformly in $L$.  Thus
\begin{eqnarray*}H_{\Bog}-E_\Bog &\geq& \frac{C}{L}SN^{>}S^{\dagger}\\&\geq& \frac{C_1}{L^2}\sum_{\pm\p\neq\0}\big(b_\p^\dagger b_\p+b_{-\p}^\dagger b_{-\p}\big)-\frac{C_2 L^{d-1}}{L^d}\sum_{\pm\p\neq\0}2s_\p(c_\p-s_\p)\\&
=&C_1L^{-2}N^>-C_2L^{d-1}.
\end{eqnarray*}
  
This proves (\ref{qe1}), which can be rewritten as
\beq N^>\leq C_1^{-1}L^{2}(H_\Bog-E_\Bog+C_2L^{d-1}),\label{qe1+}\eeq
which implies
(\ref{qe1a}). \qed

\begin{lemma}
Set
\begin{eqnarray*}
M&:=&\sum_{\p\neq\0}(c_\p-s_\p)^2b_\p^\dagger b_\p,\\
A_1&:=&\sum_{\p\neq\0}2s_\p(c_\p-s_\p),\\
A_2&:=&\sum_{\p\neq0}4(c_\p-s_\p)^2s^2_\p.\end{eqnarray*}
Then
\beq
(SN^>S^\dagger+A_1)^2\geq M^2-A_2.\eeq
\label{lemik}
\end{lemma}

\proof
\begin{eqnarray}&&\Bigg(\sum_{\pm\p\neq\0}\big(
d_\p^\dagger d_\p+d_{-\p}^\dagger d_{-\p}+2s_\p(c_\p-s_\p)\big)\Bigg)^2\notag\\
&=&
\sum_{\pm\p,\pm \q\neq\0}\Big(
d_\p^\dagger \big(d_\q^\dagger d_\q+d_{-\q}^\dagger
d_{-\q}+2s_\q(c_\q-s_\q)\big)d_\p\notag\\&&\ \ 
+d_{-\p}^\dagger \big(d_\q^\dagger d_\q+d_{-\q}^\dagger
d_{-\q}+2s_\q(c_\q-s_\q)
\big)d_{-\p}\Big)\notag
\\
&&+\sum_{\pm\p,\pm \q\neq\0}2s_\p(c_\p-s_\p)\Big(d_\q^\dagger
d_\q+d_{-\q}^\dagger d_{-\q}+2s_\q(c_\q-s_\q)\Big)\notag\\
&&+
\sum_{\pm\p\neq\0}\big(
d_\p^\dagger d_\p+d_{-\p}^\dagger d_{-\p}\big).
\notag
\end{eqnarray}
Using \eqref{sns} we bound this from below by
\begin{eqnarray}&&\sum_{\pm\p,\pm \q\neq\0}(c_\q-s_\q)^2\Big(d_\p^\dagger
 \big(b_\q^\dagger b_\q+b_{-\q}^\dagger b_{-\q}\big)d_\p+
 d_{-\p}^\dagger\big(b_\q^\dagger b_\q+b_{-\q}^\dagger b_{-\q}\big)d_{-\p}\Big)\notag \\
&&+\sum_{\pm\p,\pm \q\neq\0}2s_\p(c_\p-s_\p)(c_\q-s_\q)^2\big(b_\q^\dagger b_\q+b_{-\q}^\dagger b_{-\q}\big)\notag
\\
&&+
\sum_{\pm\p\neq\0}\Big((c_\p-s_\p)^2\big(
b_\p^\dagger b_\p+b_{-\p}^\dagger b_{-\p}\big)-2s_\p(c_\p-s_\p)\Big).
\notag\\
&=&\sum_{\pm\p,\pm \q\neq\0}(c_\q-s_\q)^2\Big(b_\q^\dagger
 \big(d_\p^\dagger d_\p+d_{-\p}^\dagger d_{-\p}\big)b_\q+
 b_{-\q}^\dagger\big(d_\p^\dagger d_\p+d_{-\p}^\dagger
 d_{-\p}\big)b_{-\p}\Big)\notag \\
&&+\sum_{\pm\p\neq\0}(c_\p-s_\p)^2\Big(
2s_\p^2\big(b_\p^\dagger b_\p+b_{-\p}^\dagger b_{-\p}\big)+2c_\p
s_\p\big(b_\p^\dagger b_{-\p}^\dagger+b_\p b_{-\p}\big)+2s_\p^2\Big)
\notag\\
&&+\sum_{\pm\p,\pm \q\neq\0}2s_\p(c_\p-s_\p)(c_\q-s_\q)^2\big(b_\q^\dagger b_\q+b_{-\q}^\dagger b_{-\q}\big)\notag
\\
&&+
\sum_{\pm\p\neq\0}\Big((c_\p-s_\p)^2\big(
b_\p^\dagger b_\p+b_{-\p}^\dagger b_{-\p}\big)-2s_\p(c_\p-s_\p)\Big).
\notag\\&=&\sum_{\pm\p,\pm \q\neq\0}(c_\q-s_\q)^2\Big(b_\q^\dagger
 \big(d_\p^\dagger d_\p+d_{-\p}^\dagger
 d_{-\p}+2s_\p(c_\p-s_\p)\big)b_\q\notag\\
&&\ \ \ \ +
 b_{-\q}^\dagger\big(d_\p^\dagger d_\p+d_{-\p}^\dagger
 d_{-\p}+2s_\p(c_\p-s_\p)\big)b_{-\p}\Big)\notag \\
&&+\sum_{\pm\p\neq\0}(c_\p-s_\p)^2\Big(
\big(2s_\p^2+1\big)\big(b_\p^\dagger b_\p+b_{-\p}^\dagger b_{-\p}\big)+2c_\p
s_\p\big(b_\p^\dagger b_{-\p}^\dagger+b_\p b_{-\p}\big)\Big)
\notag
\\
&&+
\sum_{\pm\p\neq\0}\big((c_\p-s_\p)^22s_\p^2-2s_\p(c_\p-s_\p)\big).
\notag\end{eqnarray}
Using \eqref{sns} one more time, we bound this from below by
\begin{eqnarray}
&&
\sum_{\pm\p,\pm \q\neq\0}(c_\q-s_\q)^2(c_\p-s_\p)^2\Big(b_\q^\dagger
 \big(b_\p^\dagger b_\p+b_{-\p}^\dagger b_{-\p})b_\q+
 b_{-\q}^\dagger\big(b_\p^\dagger b_\p+b_{-\p}^\dagger b_{-\p}\big)b_{-\q}\Big)\notag\\
&&+\sum_{\pm\p\neq\0}(c_\p-s_\p)^2(2s_\p^2-2c_\p s_\p+1)
\big(b_\p^\dagger b_\p+b_{-\p}^\dagger b_{-\p}\big)\notag\\
&&+\sum_{\pm\p\neq\0}\big((-2c_\p
s_\p+2s_\p^2)(c_\p-s_\p)^2-2s_\p(c_\p-s_\p)\big)\notag\\
&=&\Bigg(\sum_{\pm\p\neq0}(c_\p-s_\p)^2\big(b_\p^\dagger
b_{\p}+b_{-\p}^\dagger b_{-\p}\big)\Bigg)^2-\sum_{\pm\p\neq0}4(c_\p-s_\p)^2 c_\p s_\p.\notag
\end{eqnarray}
\qed

\begin{lemma}
There exist $C_1,C_2$ such that
\beq \big(H_\Bog-E_\Bog\big)^2\geq C_1L^{-4}(N^>)^2-C_2L^{2d-2}.\label{qe2}\eeq
Therefore,
\beq 
\one_\kappa^\Bog ( N^>)^2\one_\kappa^\Bog \leq CL^{4}(L^{d-1}+\kappa)^2.
\label{qe2a}\eeq
\end{lemma}

\proof As in the proof of Lemma \ref{lemi1},
\begin{eqnarray}\big(H_{\Bog}-E_\Bog\big)^2&\geq&
\frac{\big(\pi\sqrt{8\hat v(\0)}\big)^2}{L^{2}}\big(SN^{>}S^{\dagger}\big)^2.\label{pqo}
\end{eqnarray}
For any $\delta>0$, Lemma \ref{lemik} implies
\[(1+\delta)(SN^>S^\dagger)^2+(1+\delta^{-1})A_1^2\geq M^2-A_2.\]
Moreover, the limits $\lim\limits_{L\to\infty}\frac{A_1}{L^d}$ and 
$\lim\limits_{L\to\infty}\frac{A_2}{L^d}$ exist.
Therefore,

\begin{eqnarray*}
\big(SN^{>}S^{\dagger}\big)^2&\geq&M^2-CL^{2d}.
\end{eqnarray*}
Using (\ref{pqo}) and $M\geq C_1L^{-1}N^>$, we easily conclude that (\ref{qe2}) holds. Hence
\beqnn (N^>)^2\leq C_2^{-1}L^{4}\big((H_\Bog-E_\Bog)^2+C_3L^{2d-2}\big)
,\eeqnn
which easily implies  (\ref{qe2a}).
\qed

Suppose now that
$G$ is a smooth nonnegative function on $[0,\infty[$ such that
\begin{equation}G(s)=
\begin{cases}
 1, & \text{if }s\in[0,\frac{1}{3}]\\
 0, & \text{if }s\in[1,\infty[.
\end{cases} \label{F}
\end{equation}
Set
\[A_N:=G(N^>/N),\ \ A_N^\nph:=\one-A_N.\]
The operator $A_N$ will serve as a smooth approximation to the projection onto the physical space. 
Set \[Y_\kappa:=\one_\kappa^\Bog A_N.\]
\begin{lemma}
We have
\[\one_\kappa^\Bog-Y_\kappa Y_\kappa^\dagger=O\big(
L^{2}(\kappa+L^{d-1})N^{-1}\big).\]
\label{lle1}\end{lemma}

\proof We have
\begin{eqnarray*}
&&\one_\kappa^\Bog-Y_\kappa Y_\kappa^\dagger\ =\ \one_\kappa^\Bog\big(1-G(N^>/N)^2\big)
\one_\kappa^\Bog \\
&=&
\one_\kappa^\Bog (N^>/N)^{1/2}
\Big(\big(1-G(N^>/N)^2\big)(N^>/N)^{-1}\Big)
(N^>/N)^{1/2}
\one_\kappa^\Bog
.\end{eqnarray*}
But
\[\|\big(1-G(N^>/N)^2\big)(N^>/N)^{-1}\|=\sup\limits_s\{|(1-G(s)^2)s^{-1}|\}<\infty,\] and by (\ref{qe1a})
\[(N^>/N)^{-1/2}
\one_\kappa^\Bog=O\big(
L(\kappa+L^{d-1})^{1/2}N^{-1/2}\big).\]
\qed

Let $0<c_0<1$. If \beq \|\one_\kappa^\Bog-Y_\kappa Y_\kappa^\dagger\|\leq c_0,\label{pqi}\eeq
then
$Y_\kappa Y_\kappa^\dagger$ is invertible on $\Ran \one_\kappa^\Bog$.
 We will denote by
$\big(Y_\kappa Y_\kappa^\dagger\big)^{-1}$ the corresponding inverse. We set
\[X_\kappa:=
\big(Y_\kappa Y_\kappa^\dagger\big)^{-1/2}.\]
On the orthogonal complement of $\Ran\one_\kappa^\Bog$ we extend it by
$0$.

 By Lemma \ref{lle1} and Condition \eqref{condit2a} with a sufficiently small $c_1$, we can guarantee that \eqref{pqi} holds with, say, $c_0\leq1/2$. Therefore, in what follows $X_\kappa$ is well defined.

\begin{lemma}
\beq \one_\kappa^\Bog-X_\kappa=O\big(
L^{2}(\kappa+L^{d-1})N^{-1}\big).\label{pqo1}\eeq
\end{lemma}

\proof
\[\|\one_\kappa^\Bog-\big(Y_\kappa Y_\kappa^\dagger\big)^{-1}\|\leq c_0(1-c_0)^{-1}\]
by the convergent Neumann series. This is $O\big(
L^{2}(\kappa+L^{d-1})N^{-1}\big)$.
This implies (\ref{pqo1}) by the spectral theorem. \qed

\begin{lemma}
\[X_\kappa[A_N,[A_N,H_\Bog]]X_\kappa=O\big(N^{-2} L^{2}(\kappa+L^{d-1})\big).\]
\label{???}\end{lemma}

\proof We have
\[[N^>,[N^>,H_\Bog]]=2\sum \hat v(\p)(b_\p b_{-\p}+b_\p^\dagger b_{-\p}^\dagger).\]
Using
\[
-b_\p^\dagger b_\p-b_{-\p}^\dagger b_{-\p}-1
\leq
b_\p b_{-\p}+b_\p^\dagger b_{-\p}^\dagger\leq b_\p^\dagger b_\p+b_{-\p}^\dagger b_{-\p}+1
\] we obtain
\[-C(N^>+L^d)\leq
[N^>,[N^>,H_\Bog]]\leq
C(N^>+L^d).\]
This implies
\beq
\Big \|(N^>+L^d)^{-1/2}\Big[N^>,[N^>,H_\Bog]\Big](N^>+L^d)^{-1/2}\Big\|\leq C. 
\label{powq}\eeq
Now we use one of the well-known methods of dealing with functions of operators, for instance, the representation
\[A_N=G(N^>/N)=\frac{1}{2\pi}\int\hat G(t)\e^{\i tN^>/N}\d t.\]
To this end note that for operators $S$ and $T$ one has 
\begin{gather*}
\left[\e^{\i tS},T\right]=\int_{0}^{t}\frac{\d}{\d u}\e^{\i uS}T\e^{-\i uS}\d u\e^{\i tS}=\int_{0}^{t}\i e^{\i uS}[S,T]\e^{\i(t-u)S}\d u  
\intertext{which together with the representation mentioned above yields}
\Big[A_N,[A_N,H_\Bog]\Big]\\
=\frac{-1}{4\pi^2 N^2}\int\d p\int_{0}^{p}\d s\int\d t\int_{0}^{t}\d u\hat{G}(p)\hat{G}(t)\e^{\i(s+u)\frac{N^>}{N}}\Big[N^>,[N^>,H_\Bog]\Big]\e^{\i(t-u+p-s)\frac{N^>}{N}}.
\end{gather*}
Therefore,
\begin{eqnarray*}
&&\|X_\kappa[A_N,[A_N,H_\Bog]]X_\kappa\|\\
&\leq&\frac{1}{4\pi^2N^2}\int\d p\int\d t |p\hat G(p) t\hat G(t)|
\|X_\kappa(N^>+L^d)^{1/2}\|\\
&&\times\|(N^>+L^d)^{-1/2}[N^>,[N^>,H_\Bog]](N^>+L^d)^{-1/2}\|
\|(N^>+L^d)^{1/2}X_\kappa\|\end{eqnarray*}
Now, by  Lemma 8.2,
\begin{eqnarray*}\|X_\kappa(N^>+L^d)^{1/2}\|
\|(N^>+L^d)^{1/2}X_\kappa\|
&=&\|X_\kappa(N^>+L^d)X_\kappa\|\\
&\leq&
C\big(\|\one_\kappa^\Bog N^>\one_\kappa^\Bog+L^d\big)\\&\leq& 2C L^2(L^{d-1}+\kappa).
\end{eqnarray*}
Besides, $\hat G$ decays fast. Thus it is enough to use (\ref{powq}) to complete the proof.

\qed

We define
\[Z_\kappa:=X_\kappa A_N=
\big(\one_\kappa^\Bog A_N^2\one_\kappa^\Bog
\big)^{-1/2}
 A_N.\]
Clearly, $Z_\kappa$ is a partial isometry with initial space $\Ran (A_{N}\one_{\kappa}^{\Bog})$ and final space $\Ran (\one_{\kappa}^{\Bog})$.

\begin{lemma}\label{lemma4}
\begin{eqnarray*}
\one_\kappa^\Bog (H_\Bog-E_\Bog) \one_\kappa^\Bog&=&Z_\kappa (H_\Bog-E_\Bog) Z_\kappa^\dagger\\
&&+ O\big(L^{2}(L^{d-1}+\kappa)\kappa N^{-1}\big)\\
&&+
 O\big(L^{2}(L^{d-1}+\kappa)N^{-2}\big)
.\end{eqnarray*}
\end{lemma}
\begin{proof}
We have
\begin{eqnarray}
\one_\kappa^\Bog (H_\Bog-E_\Bog) \one_\kappa^\Bog
&=&
\big(\one_\kappa^\Bog -X_\kappa\big)
(H_\Bog-E_\Bog) \one_\kappa^\Bog\label{tag1}
\\
&&+X_\kappa
(H_\Bog-E_\Bog) 
\big(\one_\kappa^\Bog -X_\kappa\big)\label{tag2}\\
&&+X_\kappa
(H_\Bog-E_\Bog) 
X_\kappa;\notag
\end{eqnarray}
\begin{eqnarray}
X_\kappa
(H_\Bog-E_\Bog) 
X_\kappa&=&-
X_\kappa A_N^\nph
(H_\Bog-E_\Bog)  A_N^\nph
X_\kappa\notag\\
&&+
X_\kappa 
(H_\Bog-E_\Bog)  A_N^\nph
X_\kappa\label{tag3}\\
&&+
X_\kappa A_N^\nph
(H_\Bog-E_\Bog)  
X_\kappa\label{tag4}\\
&&+
X_\kappa A_N
(H_\Bog-E_\Bog)  A_NX_\kappa;\notag
\end{eqnarray}
\begin{eqnarray}
-X_\kappa A_N^\nph
(H_\Bog-E_\Bog)  A_N^\nph
X_\kappa
&\!\!\!\!{=}&\!\!\!\!\!\!
{-}\frac12
X_\kappa (A_N^\nph)^2
(H_\Bog-E_\Bog) 
X_\kappa\label{tag5}\\
&&\!\!\!\!\!\!{-}\frac12X_\kappa 
(H_\Bog-E_\Bog)  (A_N^\nph)^2
X_\kappa\label{tag6}\\
&&\!\!\!\!\!\!{+}\frac12
X_\kappa \big[A_N^\nph,[A_N^\nph,
H_\Bog]\big]
X_\kappa.\label{tag7}
\end{eqnarray}
The error term in the lemma equals the sum of (\ref{tag1}),...,(\ref{tag7}).
By (\ref{pqo1}),
\[ (\ref{tag1}) ,(\ref{tag2})= O\big(L^{2}(L^{d-1}+\kappa)\kappa N^{-1}\big).\]
By (\ref{qe2a}),
\[(\ref{tag3}),\dots, (\ref{tag6})= O\big(L^{2}(L^{d-1}+\kappa)\kappa N^{-1}\big).\]
By Lemma \ref{???},\[ (\ref{tag7})= O\big(L^{2}(L^{d-1}+\kappa) N^{-2}\big).\]
\end{proof}

\begin{lemma}\label{lemma89}
Assume (\ref{condit2}). Then
\begin{eqnarray}
\inf_{0<\epsilon\leq1}Z_\kappa R_{N,\epsilon} Z_\kappa^{\dagger}&\leq&
C L^{d/2+3}(L^{d-1}+\kappa)^{3/2}N^{-1/2}
\label{epsilo}
\end{eqnarray}
\end{lemma}

\begin{proof}
\begin{eqnarray*}
Z_\kappa R_{N,\epsilon} Z_\kappa^{\dagger}&\leq&
Z_\kappa \frac{1}{2}\sum_{\p\neq\0}\hat v(\p)\Big(\Big(\frac{\sqrt{(N_0-1)N_0}}{N}
-1\Big)\b_\p \b_{-\p}
+\hc\Big)  Z_\kappa^{\dagger}\notag\\&&+Z_\kappa \frac{\hat v(\0)}{2N}N^> Z_\kappa^{\dagger}\notag\\
&&+\epsilon Z_\kappa \frac{1}{N}\sum_{\p\neq\0}\big(\hat v(\p)+\hat v(\0)\big)\b_\p^\dagger \b_\p N_0^\ext Z_\kappa^{\dagger} \notag\\
&&+(1+\epsilon^{-1})Z_\kappa 
\frac{1}{2N} v(\0)L^d N^>(N^>-1) Z_\kappa^{\dagger}\notag\\
&\leq&\one_\kappa^\Bog C\frac{(N^>)^2+1}{N}\one_\kappa^\Bog
\notag \\
&&+\one_\kappa^\Bog C\frac{N^>}{N}\one_\kappa^\Bog \notag \\
&&+\epsilon\one_\kappa^\Bog CN^>\one_\kappa^\Bog \notag\\
&&+(1+\epsilon^{-1})\one_\kappa^\Bog C\frac{L^d(N^>)^2}{N}\one_\kappa^\Bog . 
\end{eqnarray*}
Using $\epsilon\leq1$, we can simplify the bound as follows:
\begin{eqnarray}
&\leq &\one_\kappa^\Bog \frac{C}{N}
{+}\epsilon\one_\kappa^\Bog CN^>\one_\kappa^\Bog 
{+}\epsilon^{-1}\one_\kappa^\Bog
C\frac{L^d(N^>)^2}{N}\one_\kappa^\Bog,
 \label{rnepsilon+1}
\end{eqnarray}
By (\ref{qe1a}) and (\ref{qe2a}).
this can be estimated by
\beqnn CN^{-1}+\epsilon C L^{2}(L^{d-1}+\kappa)+
\epsilon^{-1} C L^{d+4}(L^{d-1}+\kappa)^2N^{-1}.
\eeqnn
Setting  $\epsilon=c^{-1/2}N^{-1/2}L^{d/2+1}(L^{d-1}+\kappa)^{1/2}$, which is less than $1$ by   Condition (\ref{condit2}),
we obtain
\beq
 CN^{-1}+ C L^{d/2+3}(L^{d-1}+\kappa)^{3/2} N^{-1/2}.
\label{rnepsilon+}\eeq
By changing $C$, the second term can obviously absorb $CN^{-1}$.
\end{proof}

\begin{proof}[Proof of Thm \ref{upper}]
$Z_\kappa$ is a partial isometry with the initial space contained in the physical space and the final projection $\one_\kappa^\Bog$. 
Therefore,
\begin{eqnarray*}
\vecsp H_N&\leq&\vecsp\Bigg( Z_\kappa^\dagger Z_\kappa H_N
 Z_\kappa^\dagger Z_\kappa\Big|_{\Ran Z_\kappa^\dagger} \Bigg)\\
&=&\vecsp \Bigg(
Z_\kappa H_N
 Z_\kappa^\dagger\Big|_{\Ran\one_\kappa^\Bog}\Bigg).
\end{eqnarray*}

\begin{eqnarray}Z_\kappa H_N
 Z_\kappa^\dagger&\leq&Z_\kappa H_{N,\epsilon}
 Z_\kappa^\dagger\notag\\
&=&
\frac{1}{2}\hat v(\0)(N-1)\one_\kappa^\Bog
+H_\Bog\one_\kappa^\Bog\notag\\
&&+Z_\kappa (H_\Bog -E_\Bog)Z_\kappa^\dagger-(H_\Bog-E_\Bog)\one_\kappa^\Bog
\label{lab1}\\
&&
+Z_\kappa R_{N,\epsilon}  Z_\kappa^\dagger.
\label{lab2}
\end{eqnarray}
By Lemma \ref{lemma4},
\begin{eqnarray}
(\ref{lab1})&\leq&
CL^{2}(L^{d-1}+\kappa)\kappa N^{-1}\label{szac1}\\
&&+
CL^{2}(L^{d-1}+\kappa)N^{-2}.\label{szac2}
\end{eqnarray}
Using $\kappa<\kappa+L^{d-1}$ and later \eqref{condit2} we have
\begin{eqnarray*}
(\ref{szac1})&\leq &CL^2(L^{d-1}+\kappa)^{2} N^{-1}\\
&\leq&CL^{-d/2+1}(L^{d-1}+\kappa)^{3/2} N^{-1/2}.
\end{eqnarray*}
Thus (\ref{szac1}) can be absorbed in
$O(L^{d/2+3}(L^{d-1}+\kappa)^{3/2} N^{-1/2})$.  

We easily check that the same is true in the case of  (\ref{szac2}).  To bound \eqref{lab2} we use Lemma \ref{lemma89}.
\end{proof}

\begin{proof}[Proof of Thm \ref{thm} (2)]
First set $\kappa=0$. Then Condition (\ref{condit2})
 becomes Condition (\ref{war2a}) and 
 Condition (\ref{condit2a})
 becomes Condition (\ref{war2a+}). 
We obtain Thm \ref{thm} (2a).

Next set $\kappa=K_\Bog^j(\p)$. Then Condition (\ref{condit2}) is equivalent to the conjunction of Conditions (\ref{war2a}) and (\ref{war2b}).
Condition (\ref{condit2a}) is equivalent to the conjunction of Conditions (\ref{war2a+}) and (\ref{war2b+}).
This shows Thm \ref{thm} (2b).
\end{proof}

\appendix
\section{Proof of Corollary \ref{coro}}

The proof of the corollary is based on the following lemma:

\begin{lemma}\label{lemmaappendix}
\begin{enumerate}
\item Let $b>1$, $-1-\frac{1}{d+1}\leq \alpha \leq 1$ and $L^{4d+6}\leq b N^{1-\alpha}$.
 Then \vspace{2ex}
\begin{enumerate}
\item \hspace{8ex} $\frac{1}{2}\hat v(\0)(N-1)+E_\Bog\leq E_N+O(N^{-\alpha/2})$;
\item if $K_N^j(\p)\leq (bN^{1-\alpha}L^{-d-6})^{1/3}$, then
\[ \frac{1}{2}\hat v(\0)(N-1)+E_\Bog +K_\Bog^j(\p)\leq E_N+ K_N^j(\p)+O(N^{-\alpha/2});\]
\item if $0\leq\alpha\leq1$ and $K_N^j(\p)\leq bN^{1-\alpha}L^{-d-6}$, then
 \begin{eqnarray*}
 \frac{1}{2}\hat v(\0)(N-1)+E_\Bog+K_\Bog^j(\p)&\leq& E_N+ K_N^j(\p)\\&&+\left(1+K_N^j(\p)\right)O(N^{-\alpha/2}).\end{eqnarray*}
\end{enumerate}
\item
Let $b>1$,  $-1-\frac{1}{2d+1}< \alpha \leq 1$ and $L^{4d+3}\leq bN^{1-\alpha}$. Then there exists $M$ such that if $N>M$, then \vspace{2ex}
\begin{enumerate}\item
\hspace{8ex}
$ E_N\leq \frac{1}{2}\hat v(\0)(N-1)+E_\Bog+O(N^{-\alpha/2})$;
\item if $K_\Bog^j(\p)\leq (bN^{1-\alpha}L^{-d-6})^{1/3}$, then
\[ E_N+K_N^j(\p)\leq \frac{1}{2}\hat v(\0)(N-1)+E_\Bog+K_\Bog^j(\p)+O(N^{-\alpha/2});\]
\item  if  $0<\alpha\leq1$ 
and  $K_\Bog^j(\p)\leq bN^{1-\alpha}L^{-d-6}$, then
\begin{eqnarray*}
 E_N+K_N^j(\p)&\leq&\frac{1}{2}\hat v(\0)(N-1)+E_\Bog+ K_\Bog^j(\p)\\
&&+\left(1+K_\Bog^j(\p)\right)O(N^{-\alpha/2}).\end{eqnarray*}
\end{enumerate}\end{enumerate}
\end{lemma}

\begin{proof}
To prove (1), resp. (2) we use Thm \ref{thm} (1), resp. (2). We give a proof of the latter part, since it is slightly more involved (because of the parameter $c_1$).

(2a): First we check  Condition (\ref{war2a}):
\begin{eqnarray}
L^{2d+1}&=&\big(L^{4d+3}\big)^{\frac{2d+1}{4d+3}}\leq \big(bN^{1-\alpha}\big)^{\frac{2d+1}{4d+3}}.
\label{noto9}\end{eqnarray}
For  $-1-\frac{1}{2d+1}\leq \alpha$ we have (\ref{noto9})$\leq cN$.

Next,
\begin{eqnarray}
L^{d+1}&=&\big(L^{4d+3}\big)^{\frac{d+1}{4d+3}}\leq\big(bN^{1-\alpha}\big)^{\frac{d+1}{4d+3}}.
\label{noto9+}\end{eqnarray}
We have (\ref{noto9+})$\leq cN N^{\frac{-3d-2-\alpha(d+1)}{4d+3}}$. Therefore, for  $-1-\frac{1}{2d+1}\leq \alpha$,  
  Condition (\ref{war2a+}) is satisfied for large enough $N$.

Then we apply Thm \ref{thm} (2)(a), using
\[N^{-1/2}L^{2d+\frac32}\leq N^{-1/2}(bN^{1-\alpha})^{1/2}=O(N^{-\alpha/2}).\]

(2b): We check  Condition (\ref{war2b}):
\begin{eqnarray}\notag
K_{\Bog}^j(\p)&\leq&\big(bN^{1-\alpha}L^{-d-6}\big)^{1/3}\\\notag
&\leq&\big(bN^{1-\alpha}L^{-d-6}\big)^{1/3}\big(bN^{1-\alpha}L^{-4d-3}\big)^{\frac{2d}{12d+9}}\\
&=&C N^{\frac{(1-\alpha)(2d+1)}{12d+9}}L^{-d-2}.\label{noto}
\end{eqnarray}
For $ -1-\frac{1}{2d+1}\leq\alpha$ we have (\ref{noto})$\leq CNL^{-d-2}$.

Also
\begin{eqnarray*}
\ref{noto}=O(N^{\frac{-(2+\alpha)(4d+3)+2d(1-\alpha)}{12d+9}}) NL^{-2-d}
\end{eqnarray*}
which implies
\begin{eqnarray}\notag
K_{\Bog}^j(\p)
&\leq& O(N^{\frac{-(2+\alpha)(4d+3)+2d(1-\alpha)}{12d+9}}) NL^{-2}
.\notag
\end{eqnarray}
Therefore, if $ -1-\frac{1}{2d+1}<\alpha$,  Condition (\ref{war2b+}) is satisfied for large enough $N$.

We clearly have
\begin{eqnarray} N^{-1/2}L^{d/2+3}\big(K_{\Bog}^j(\p)+L^{d-1}\big)^{3/2}&\leq&
2^{3/2}N^{-1/2}L^{d/2+3}K_{\Bog}^j(\p)^{3/2}\label{noto0}\\&&+2^{3/2}N^{-1/2}L^{2d+\frac32}.
\label{noto1}\end{eqnarray} 
We already know that   (\ref{noto1}) is $O(N^{-\alpha/2})$. Thus to apply Thm \ref{thm} (2b) we need only to bound (\ref{noto0}):
\[N^{-1/2}L^{d/2+3}K_{\Bog}^j(\p)^{3/2}\leq
N^{-1/2}L^{d/2+3}\big(bN^{1-\alpha}L^{-d-6}\big)^{1/2}=O(N^{-\alpha/2}).
\]

(2c):  Condition (\ref{war2b}) is trivially satisfied, since for $L\geq1$, $N\geq1$ and $\alpha > 0$
\[K_{\Bog}^j(\p)\leq bN^{1-\alpha}L^{-d-6}\leq bNL^{-d-2}.
\]

We have
\begin{eqnarray}\notag
K_{\Bog}^j(\p)
&\leq&bN^{1-\alpha}L^{-d-6}
L^{d+4}\\
&=&O(N^{-\alpha}) NL^{-2}
.\notag
\end{eqnarray}
Therefore Condition (\ref{war2b+}) is satisfied for large enough $N$.

To apply Thm \ref{thm} (2b) we  bound  (\ref{noto0}):
\begin{eqnarray*}
N^{-1/2}L^{d/2+3}K_{\Bog}^j(\p)^{3/2}&=&
b^{1/2}N^{-\alpha/2}\big(b^{-1}N^{-(1-\alpha)}L^{d+6}\big)^{1/2}K_{\Bog}^j(\p)^{3/2}\\
&\leq& O(N^{-\alpha/2})K_{\Bog}^j(\p).\end{eqnarray*}
\end{proof}

\noindent{\em Proof of Corollary \ref{coro}} Part (1) follows directly from Lemma \ref{lemmaappendix} (1a) and (2a).

Let us prove (2). To simplify notation we drop $\p$ from
 $K_N^j(\p)$ and $K_\Bog^j(\p)$.

Assume first that $K_N^j\leq K_\Bog^j$. By Lemma \ref{lemmaappendix} (1b) for some $C>0$ 
\begin{eqnarray*}
\frac{1}{2}\hat v(\0)(N-1)+E_\Bog+K_\Bog^j&\leq &E_N+K_N^j+CN^{-\alpha/2}\\
&\leq &E_N+K_\Bog^j+CN^{-\alpha/2}.\end{eqnarray*}
Thus
\begin{eqnarray*}\frac{1}{2}\hat v(\0)(N-1)+E_\Bog-E_N+K_\Bog^j-CN^{-\alpha/2}&\leq& K_N^j\leq K_\Bog^j.
\end{eqnarray*}
By Lemma \ref{lemmaappendix} (2a), \[-CN^{-\alpha/2}\leq
\frac12 \hat{v}(\0)(N-1)+E_\Bog-E_N.\]
Hence the statement follows.

Assume now that $K_\Bog^j\leq K_N^j$. Then we use Lemma \ref{lemmaappendix} (2b) and obtain
\begin{gather*}
K_\Bog^j\leq K_N^j\leq \frac{1}{2}\hat v(\0)(N-1)+E_\Bog-E_N+K_\Bog^j+CN^{-\alpha/2}.
\end{gather*}
By Lemma \ref{lemmaappendix} (1a),
\[\frac{1}{2}\hat v(\0)(N-1)+E_\Bog-E_N\leq CN^{-\alpha/2}.\]
 The statement follows again. This ends the proof of part (2).

The proof of part (3) is similar, except that one  uses Lemma \ref{lemmaappendix} (1c) and (2c).
\qed 
\\
\\

\textbf{Acknowledgments.} The research of J.D. and M.N was supported in part by the National Science Center (NCN) grant No. 2011/01/B/ST1/04929. The work of M.N. was also supported by the Foundation for Polish Science International PhD Projects Programme co-financed by the EU European Regional Development Fund and the National Science Center (NCN) grant No. 2012/07/N/ST1/03185.

We thank anonymous referees for their useful comments.

\end{document}